\documentclass[10pt, journal]{IEEEtran}

\ifCLASSOPTIONcompsoc
  \usepackage[nocompress]{cite}
\else
  \usepackage{cite}
\fi
\usepackage{cite,url}
\usepackage{booktabs}
\usepackage{amsmath,amssymb,bm,dsfont}
\usepackage{algorithmic}
\usepackage{amsthm,nccmath,lipsum,cuted}
\usepackage[dvips]{color}
\usepackage{graphicx}
\usepackage[lined,boxed,commentsnumbered, ruled]{algorithm2e}
\usepackage{subcaption}
\usepackage{cleveref}
\usepackage{amsfonts}
\usepackage{ragged2e}
\newtheorem{theorem}{Theorem}
\newtheorem{lemma}{Lemma}

\newtheorem{corollary}{Corollary}

\usepackage{graphicx}
\usepackage[outdir=./]{epstopdf}
\usepackage{verbatim}
\usepackage[footnote,draft,danish,silent,nomargin]{fixme}

\def\eps{\varepsilon}
\def\P{{\mathbb P}}  
\def\E{{\mathbb E}}  
\def\S{{\cal S}}
\def\A{{\cal A}}
\def\D{{\cal D}}
\def\B{{\cal B}}
\def\W{{\mathcal W}}
\def\M{{\mathcal M}}

\def\X{{\mathcal X}}
\def\Y{{\mathcal Y}}

\def\K{{\mathcal K}}
\def\INTGR{{\mathds Z}}   

\def\mx{{$(\min, \times)$}}

\newcommand{\Kd}{\mathcal{K}}

\newcommand{\EE}{\mathbb{E}}

\graphicspath{{./figs/},{../eps/}}

\begin{document}

\title{Transient Delay Bounds for Multi-Hop Wireless Networks}
\author{ Jaya Prakash Champati, Hussein Al-Zubaidy, James Gross\\ 
         School of Electrical Engineering and Computer Science, KTH Royal Institute of Technology, Stockholm, Sweden. \\
        E-mail: \{jpra,hzubaidy\}@kth.se, james.gross@ee.kth.se.
        }%


\maketitle

\begin{abstract}
In this article, we investigate the transient behavior of a sequence of packets/bits traversing a multi-hop wireless network.
Our work is motivated by novel applications from the domain of process automation, Machine-Type Communication (MTC) and cyber-physical systems, where short messages are communicated and statistical guarantees need to be provided on a per-message level.
In order to optimize such a network, apart from understanding the stationary system dynamics, an understanding of the short-term dynamics (i.e., transient behavior) is also required.
To this end, we derive novel Wireless Transient Bounds (WTB) for end-to-end delay and backlog in a multi-hop wireless network using stochastic network calculus approach. 
WTB depends on the initial backlog at each node as well as the instantaneous channel states.     
We numerically compare WTB with State-Of-The-Art Transient bounds (SOTAT), that can be obtained  by adapting existing stationary bounds, as well as simulation of the network. 
While SOTAT and stationary bounds are not able to capture the short-term system dynamics well, WTB provides relatively tight upper bound and has a decay rate that closely matches the simulation.
This is achieved by WTB only with a slight increase in the computational complexity, by a factor of $O(T+N)$, where $T$ is the duration of the arriving sequence and $N$ is the number of hops in the network. We believe that the presented analysis and the bounds can be used as base for future work on transient network optimization, e.g., in massive MTC, critical MTC, edge computing and autonomous vehicle.
\end{abstract}

\vspace{2mm}
\begin{IEEEkeywords}
Transient analysis; machine type communication; stochastic network calculus; time-critical networks; wireless
\end{IEEEkeywords}

 
\section{Introduction}
With the advent of new applications from the automation domain, it is commonly accepted that wireless networks are facing significant design challenges with respect to new quality-of-service demands.
In contrast to wireless networks optimized for human-related applications (like voice or mobile phone apps), there is a much stronger emphasis on strict reliability guarantees with respect to rather short deadlines that need to be fulfilled during the operation of the network.
This is in particular true for Ultra Reliable Low Latency Communication (URLLC) applications proposed for fifth generation (5G) cellular networks.
Such applications typically have reliability requirements with respect to acceptable packet error rates in the order of $10^{-9}$ \cite{Yilmaz-EricssonResearch}, while the end-to-end delays may not exceed a few milliseconds \cite{Dahlman-EricssonResearch}.
These requirements are in contrast to typical latency/reliability requirements of human-related applications; for example, voice applications require an end-to-end delay of about one hundred milliseconds.

To fulfill such strict requirements, one fundamental challenge is the development of network models, and subsequent network optimization algorithms and protocols, that target end-to-end performance over very short time spans.
Of particular interest is the end-to-end delay of a given sequence of bits/packets transported over a multi-hop route as it is closely related to the performance and/or safety of a control application that the data belongs to.
For instance, consider an event-based closed-loop control where a packet with a new event is instantaneously generated.
It is well known that latency in delivery of the packet caused due to underlying communication system  has a significant impact on the control performance, i.e. the stability of the controlled plant.
Likewise, safety-constrained control systems demand a periodic successful transmission of so-called `keep alive' packets to validate the system conditions, otherwise they transit into fail-safe state \cite{SafetyRequirements}.
In both cases, the involved  latency requirements can be quite small, emphasizing the necessity for a precise understanding of the short-term end-to-end communication performance.
In other words, the question thus arises on a per packet/message level rather than the traditional per flow level - with what likelihood a packet will be received (possibly in a multi-hop setting) given a deadline by the control application?

\vspace{.1cm}
Given a model that facilitates the analysis of the end-to-end delay on a packet level one could potentially design algorithms that instantaneously redistribute network resources to better accommodate, for instance, the given safety requirements.
However, achieving this rests on understanding the short-term stochastic fluctuations of the latency of a given route. 
To this end, two aspects need to be taken into account: 
(i) the instantaneous backlog which influences the end-to-end performance of a newly injected, time--critical packet sequence, and 
(ii) the short-term variability of the upcoming wireless service that results in random service increments. 
In combination, these two aspects call for approaches that can account for the queuing  as well as fading channel fluctuation effects as a starting point for any further network optimization/management tasks. {\color{black} To this end, we consider a new queueing-theoretic model where the queues have non-zero initial backlog and random service processes, and study its short-term behaviour.} 


Queuing systems operate in one of two states, \textit{transient state} or \textit{steady state}. 
It is well known that under certain conditions - for instance, stationarity of the arrival and service process as well as system stability - a queuing system, starting initially in transient state, reaches its steady state after some elapsed time.
The steady state is characterized by stationarity of the metrics of the system such as the virtual delay/sojourn time as well as the queue length.
In contrast, if these metrics are not stationary (yet), i.e. the distribution of the virtual delay for instance changes from time slot to time slot, then  the system is still operating in transient state.
Keeping this in mind, 
our focus is on the analysis of the short-term virtual delay of a tandem of queueing systems given the initial backlog of the system.
Due to the dependency of the virtual delay on the initial state and the rather short time spans that we are interested in, the virtual delay attains a non-stationary distribution, which motivates us to refer to our analysis in the following as \textit{transient} analysis.


The literature on transient network analysis is generally sparse compared to the queuing-theoretic literature body on stationary/steady state metrics.
This is due to the fact that transient queuing analysis pursued so far quickly becomes intractable.
For example, while for simple M/M/1 or more general Markovian queuing systems the steady state is governed by the (conceptually simple) flow balance equations, in case of transient analysis the involved differential equations lead to intractable expressions even for M/M/1 systems~\cite{Book:morse}.
Subsequently, either approximations or numerical methods for the transient analysis have been proposed~\cite{Zhang_91,Matis_01,Czachorski_15}.
Furthermore, despite the relevance of transient analysis for communication networks, it has received little attention when analyzing practically relevant networking effects. 
One exception is~\cite{Mellia_02}, where   the selection of the TCP congestion window is studied by applying transient analysis for flows of short lengths. 
Through a simple recursive formula for the average completion time of the flow transmission, the authors showed a significant impact of different window settings.
Nevertheless, their model does not account for queuing effects along the route, among many other aspects. 
A second application example is the analysis of ATM networks~\cite{wang_96}, where   a transient analysis based on an extension of Petri nets is presented.
While demonstrating a very strong aspect of transient analysis in general - for example, the ability to characterize practically relevant overload situations (which cannot be dealt with using steady state analysis) - the presented approach nevertheless rests heavily on numerical methods that limit the analytical insight. Moreover, incorporating service elements based on the fading  distribution of wireless channel in these models pose a significant challenge for their tractability.

Alternative approaches to stationary (wireless) network performance analysis comprise effective capacity and stochastic network calculus.
Effective capacity was initially devised to provide asymptotic delay and backlog performance, i.e., for values of delay/backlog going to infinity, for a Rayleigh fading channel~\cite{wu}. 
The approach was later tailored towards analyzing stationary performance metrics, but can only provide asymptotic results which renders it useless for time-critical network analysis. 
Stochastic network calculus~\cite{jiang,fidler_mgf,AlZubaidyTON}, on the other hand, has been employed in the literature to study the (non-asymptotic) stationary performance
of wireless networks; see~\cite{alzubaidyInfocom13,itc14,icc15,jiang_markov,florin_infocom,florin_capacity,fidler_globecom}. 
Nevertheless, \textit{deriving performance bounds for transient behavior of wireless networks} has not been attempted before. 
One exception is the work by Becker and Fidler~\cite{Becker-Fidler-ITC2015}, where they analyzed the effects of transient phases on delay and backlog for sleep scheduling in wireless networks by proposing non-stationary service curves. However, they numerically evaluate the probabilistic bounds for specific service processes, and do not consider the effect of initial backlog in the system.

In this work, we strive to establish transient network performance analysis, in particular, for wireless networking.
To this end, we propose an analytic model based on the stochastic network calculus framework, and study different bounding techniques that allow an accurate characterization of the probabilistic delay of a multi-hop wireless network in transient settings. 
Note that this model must capture the statistical behavior of the fading channel as well as the queuing effect at every intermediate node. 
This distinguishes our proposed model from most communication theoretical models that emphasize channel characterization and ignore queuing effects, and from traditional queuing models that abstract the physical channel behavior, e.g., by using  Markov channel models.

The main contribution of this work lies in deriving new performance bounds for single- and multi-hop transmission scenarios that can be used to characterize the probabilistic end-to-end delay of a mission-critical packet/message transmission. The proposed bounds do not rely on any restriction on the arrival process and is applicable to i.i.d. wireless channels with general distribution. Also, they are shown to be significantly tighter (numerically) compared with the available stationary bounds in the literature. This is achieved by utilizing information pertaining to the short-term behavior of the system caused by abruptly arriving messages, e.g., control messages in process automation,
 traversing the wireless network. Further, we derive probabilistic bounds for the backlog in the system, and present a method to extend the bounds for the case where initial backlog information is delayed.

The rest of the paper is organized as follows. 
In Section~\ref{sec:model}, we describe the network model and state basic assumptions. 
In Section~\ref{sec:Preliminaries}, we provide some background for the problem and the used methodology.  
In Section~\ref{sec:TransientBound}, we present our main contribution, namely the derivation of novel performance bounds for transient network operation. 
In Section~\ref{sec:numerical}, we present numerical results to evaluate the proposed bounds and compare them to simulation results for multi-hop wireless networks. 
We conclude in Section~\ref{sec:conclusions}.

\section{Network Model} 
\label{sec:model}
 
\begin{figure}
\centering
\includegraphics[scale=.45]{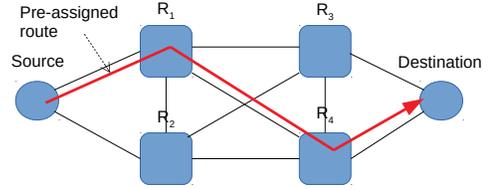}
\caption{Example of wireless network topology.}
\label{fig:NetTopology}
\vspace{-.3cm}
\end{figure}

We investigate a wireless network composed of a source, a destination and multiple relay nodes in a mesh topology; see for example Figure \ref{fig:NetTopology}. 
Delay-critical packets (generated by a control application) are exchanged between the source and the destination  and are constrained by a hard end-to-end deadline $w$.
Each pair of nodes are connected via point-to-point fading channel with link data rate equivalent to the channel capacity. 
We assume a static routing that pre-assigns a route for the flow of interest between the source and destination. 
The network topology then is modeled as an undirected graph, where the set of vertices represent the nodes in the network and the set of edges represent the physical links between the nodes.  
A route is an ordered set of links connecting the source and the destination as shown by the red line in Figure \ref{fig:NetTopology}. 
Our goal is to develop a model that facilitates end-to-end delay analysis on a per-message basis for the time-critical data over a given route through the graph.
Towards that end, we next introduce a more formal model for the multi-hop route considering in particular the tandem of queuing systems.
Afterwards, we give a precise problem statement.

\subsection{Queuing Model}
We consider a fluid-flow, discrete-time queuing model for a multi-hop wireless network shown in Figure~\ref{fig:SysDesc}. 
A data flow traverses the $N$-hop wireless links. 
We are interested in studying the performance of  message abruptly arriving to this network at time $t_0$ and lasting for  $T$ time slots, where $t_0 \geq 0$ is any arbitrary time. 
This message may represent a sequence of time-critical data bits (or packets) arriving at time $t_0$ and lasting for $T$ time slots. 
At time $t$, we denote the buffer state of the network by $\mathbf{x}(t)\in \INTGR_+^N $, where $x_n(t)$ is the  backlog of wireless link $n$ at time $t$,  $n \in \{1,2,\ldots,N\}$ and $\INTGR_+$ is the set of positive integers. To simplify notation, we designate the initial backlog at time $t_0$ as $\mathbf{x}(t_0)\equiv  \mathbf{x}$. 
Given $\mathbf{x}$, we ignore the arrivals before time $t_0$ and simply consider that the system started at time $t_0$ with initial backlog $\mathbf{x}$. 
Later, in Section~\ref{subsec:backlog} we also consider the case where $\mathbf{x}$ is not given, but delayed backlog information is available, e.g., through measurements at time $t_0 - d$ for $0\le d \le t_0$.

The arrivals of our interest is described by the cumulative arrival process $A(t_0,t)$, $0\le t_0 \le t$ where  $A(u,t)= \sum_{i=u}^{t-1} a_i$, for all $  t_0\le u \le t$, $a_i$ is the arrival increment during time slot $i$ and $a_i =0$, for all $i\notin [t_0, t_0 + T)$. After being served by the system, the arrivals result in a departure process $D(t_0,t)$.
Given $t_0$, we define $A(t) = A(t_0,t)$ and $D(t) = D(t_0,t)$. {\color{black} In this work, we do not impose any restriction on $A(t)$ - arrival increments can take independent and arbitrary values - in deriving the bounds. We also present expressions for the case where $A(t)$ obeys the $(\sigma,\rho)$-bounded traffic characterization introduced by Cruz~\cite{Cruz1991} --  a typical assumption for deriving bounds in the network calculus literature. Under $(\sigma,\rho)$-bounded traffic characterization,
\begin{equation}\label{eq:boundedarrival}
A(t) - A(u) \leq \rho(t-u) + \sigma, \; \forall \; t_0 \leq u \leq t,
\end{equation}
for some $\sigma \geq 0$ and $\rho \geq 0$.}

The cumulative service provided by the $n^{\rm th}$ wireless link is given by 
\begin{equation}\label{eq:Service-Process}
S_n(u,t) = \sum_{i = u}^{t-1} s_{n,i},
\end{equation} 
where $s_{n,i}$ is the capacity of the $n^{\rm th}$ wireless link during time slot $i$.
We assume that the service processes are i.i.d. both across links and time slots. {\color{black} We consider general distribution for the service processes and derive the results, but for the purposes of numerical evaluation and simulation we use Rayliegh-block fading channel model.} 
We assume first-come first-serve discipline for the designated arrival sequence at each store-and-forward node along the path. 


\begin{figure}
\centering
\includegraphics[scale=.45]{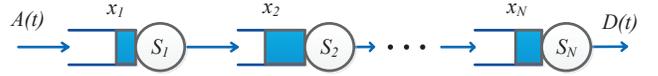}
\caption{ Multi-hop wireless network queuing model.}
\label{fig:SysDesc}
\end{figure}

The total backlog $B(t)$ and the end-to-end virtual delay $W(t)$ of the queuing network described above is then given by 
\begin{equation}\label{eq:backlog}
B(t) = A(t) +\sum_{n=1}^N x_n - D(t),
\end{equation}  
and 
\begin{equation} \label{eq:delay}
W(t) = \inf \left \{w \ge 0 : A(t) +\sum_{n=1}^N x_n \le D(t+w)\right \}.
\end{equation} 


\subsection{Problem Statement}
As mentioned earlier, the arriving traffic of interest $A(t)$ is constrained by the hard deadline $w$, i.e. data transmitted at time slot $t$ must be received before $t + w + 1$.
In order to efficiently support such constrained traffic $A(t)$, we strive for an analysis of the \textit{deadline violation probability} of the virtual delay $W(t)$ for $t \in \left[t_0,t_0+T\right]$, i.e. an analysis of $\P(W(t) > w)$ for the time-critical data as it traverses the multi-hop wireless network described above under the crucial assumption of the initial backlog $\mathbf{x}$ along the route.

It is well known to the queueing-theoretic community that under certain conditions (typically referred to as stability criteria) the stochastic properties of the metrics such as virtual delay/sojourn time settle, turning it into a stationary random process \cite{Book:chang}.
This state of the queueing system is also referred to as steady state. In contrast, the case where the stochastic properties have not settled yet is commonly referred to as transient state.
In this spirit, we refer to the analysis goal discussed above as transient analysis, as the stochastic properties of the virtual delay of interest have not settled yet.
As we will see, our main contributions lie in the derivation of new bounds for various cases, and consequently we will refer to them also as \textit{transient bounds}.
Alternatively, one might also consider the corresponding stationary case, for which state of the art readily provides so called \textit{stationary bounds}.
However, these stationary bounds provide ample room for improvement, resulting from the short-term aspect of the virtual delay of interest as well as the fact that the initial backlog is known.
Finally, for convenience of exposition we define $\tau = t + w$ and $x_\text{max} = \max\limits_{n\in \{ 1, \dots, N\} } x_n$.
\section{Metodology and Existing Results}\label{sec:Preliminaries}
The transient analysis of network performance provides a better understanding of the quality of service requirements (and hence, an efficient delivery) of \emph{spontaneous traffic} traversing multi-hop data network. 
This analysis is particularly important for mission-critical communications over wireless networks, arising from scenarios such as industrial IoT, cyber-physical systems or vehicular applications, which all require the timely delivery of information with high reliability. 
Hence, the problem at hand is to evaluate the transient performance (in terms of end-to-end delay and reliability) of mission-critical traffic traversing a multi-hop wireless network.  
One possible approach for the performance analysis of such networks is based on stochastic network calculus theory \cite{jiang,fidler_mgf}  and its recently reported application to wireless networks analysis~\cite{AlZubaidyTON}. 
The key benefit from using stochastic network calculus is the ability to extend single hop results to multi-hop settings with reasonable efforts. 
Furthermore, the described approach provides closed-form expressions in terms of the physical attributes of the wireless fading channel. 
Nevertheless, network calculus in general provides bounds rather than exact results, which is a necessary compromise to achieve tractability.
Note that in the context of mission-critical transmissions upper bounds on the transmission reliability with respect to a given latency target are generally acceptable for network design and management purposes.

In the literature, there are several approaches, based on stochastic network calculus, that can handle the performance analysis of wireless networks. 
These approaches range from computing the effective capacity of such channels \cite{wu} to computing the MGF of a Markov process abstraction of the wireless fading channel \cite{fidler_globecom} and ON-OFF service characterization of slotted Aloha access over shared wireless channel  \cite{florin_capacity}. 
A more recent approach, namely \mx~network calculus, that provides end-to-end performance characterization of wireless networks in terms of the fading channel physical parameters, i.e., fading distribution and average SNR, is developed in \cite{AlZubaidyTON} and is based on \mx~dioid algebra. 
In this paper, we pursue the transient analysis of wireless systems by utilizing the \mx~network calculus, while we note that in principle this could also be pursued for example by MGF-based network calculus.
  
\subsection{Equivalent Model}  
In order to analyze the network in Figure~\ref{fig:SysDesc} using stochastic network calculus approach, we must first overcome the following incompatibility: 
In the network calculus framework it is assumed that initially all buffers are empty and no arrivals (from the considered traffic stream) has happened before the start time $t_0$, i.e., $\sum_{n=1}^N x_n =0$ and $A(0,t_0) = 0$. 
Furthermore, it is assumed that no service is rendered by time $t_0$, i.e.,  $S(0,t_0) = 0$. 
To this end, we define an alternative, yet equivalent, queuing model for our system shown in Figure~\ref{fig:SysModel}. 
We treat the initial backlog $x_n$ at link $n$ as cross-traffic, $A^c_n(t)$, given by
\begin{align}\label{eq:InitBacklog}
A^c_n(t)= \min (\kappa(t-t_0) , x_n),
\end{align}
where $\kappa(t)$ is a burst function with $\kappa(t)= 0$, for $t = 0$, and $\kappa(t) = \infty$, otherwise. The devised model satisfies the requirements  for a network-calculus-based analysis. 

For the ease of exposition we use $A_n(t)$ and $D_n(t)$ to denote the cumulative arrivals and cumulative departures, respectively, at link $n$. Note that $ A_1(t) = A(t) + A^c_1(t)$, 
\begin{equation}
A_n(t) = D_{n-1}(t) + A^c_n(t), \; \forall n \in \{ 2,\dots, N\}
\end{equation}
and $D(t) = D_N(t)$. 

Given the equivalent model, without loss of generality, we let $t_0 = 0$ in the rest of the paper. 
In this case, we are interested in the network performance during the period $0 \le t \le T$.

\begin{figure}
\centering
\includegraphics[width = 3.4in]{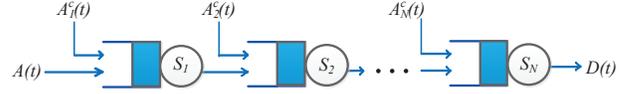}
\caption{Equivalent model.}
\label{fig:SysModel}
\end{figure}

\subsection{\mx~Network Calculus for Wireless Network Analysis}
\label{sec:Min_x_algebra}
The main objective of \mx~network calculus is to obtain probabilistic performance bounds  for multi-hop wireless networks in terms of the underlying fading channel parameters. 
A key concept of the \mx~network calculus is the transformation of the system model into an alternative analysis domain, known as \textit{SNR domain},  using the exponential function. 
In this domain, the random service rendered by a wireless fading channel is characterized in terms of the variability of the instantaneous SNR, that is the SNR service process at wireless link $n$ is given by
\begin{equation}\label{eq:SNR-Service-Process}
\S_n(u,t) = e^{S_n(u,t)} = \prod_{i = u}^{t-1}  e^{s_{n,i}} \, ,
\end{equation}  
where we use the calligraphic font to represent corresponding processes in the SNR domain.
Similarly, the cumulative arrivals and departures in the SNR domain are given by
\begin{align*}
& \A_n(u,t) = e^{A_n(u,t)} = \D_{n-1}(u,t) \cdot \A^c_n(u,t), 
\end{align*} 
where
\begin{align*}
&  \A^c_n(u,t) = \min(e^{\kappa(t-u)},e^{x_n}),  \text{ and } \, \D_n(u,t) = e^{D_n(u,t)}\, .
\end{align*} 
 Then using \eqref{eq:backlog}, the SNR backlog process is described by 
\begin{equation*}
\B(t) = e^{B(t)} = \frac{\A(t)}{\D(t)} \cdot \prod_{n=1}^{N} e^{x_n}.
\end{equation*}
However, the transformation does not affect time. Therefore the delay in the SNR domain is given by 
\begin{equation*}
\W(t) = W(t) = \inf \left \{w \ge 0 : \A(t) \cdot \prod_{n=1}^N e^{x_n} \le \D(t+w)\right \}.
\end{equation*}

The equivalent input/output relationship in \mx-algebra is given by $\D(0,t) \geq \A \otimes \S \left(0,t\right)$, where $\otimes$ is the \mx-convolution operator defined as
\begin{equation*}
 \X \otimes \Y \left(u,t\right) = \inf_{u \leq v \leq t} \left\{ \X(u,v)\cdot \Y (v,t) \right\} \, .
 \label{eq:convolution}
 \end{equation*}

Performance analysis of communication networks  often focuses on a stochastic characterization of virtual delay.
As shown in~\cite{AlZubaidyTON} an upper bound for the delay violation probability can be derived in terms of an integral transform, namely, the Mellin transform of the cumulative arrival and service processes in the SNR domain and by using the moment bound. The Mellin transform of a non-negative random process $\X$ is defined as
\begin{equation}
\mathcal{M}_{\mathcal{X}}\left(s,u,t\right) = \mathcal{M}_{\mathcal{X}\left(u,t\right)} \left(s\right) = \EE\left[\mathcal{X}^{s-1}\left(u,t\right)\right],
\label{eq:Mellin_Definition}
\end{equation}
for any $s \in \mathbb{R}$, whenever the expectation $\E[\cdot]$ exists. 
We restate some relevant results from \cite{AlZubaidyTON} in the following theorem.
\begin{theorem}\label{thm:SOTABound}
A probabilistic bound for the virtual delay at time $t$ is given by $\P (\W(t) > w^\eps) \leq \eps$, where $w^\eps$ is the smallest $w\ge 0$ that satisfies
\begin{equation}\label{eq:delay-bound}
\inf_{s>0} \left\lbrace \K(s,\tau,t) \right\rbrace \le \eps \, ,
\end{equation}
where $\tau = t +w$ and
 \begin{equation}
\label{eq:function_M_Hussein}
  \Kd(s, \tau, t) = \sum_{u=0}^{t} \mathcal{M}_{\mathcal{A}}(1+s,u,t) \cdot \mathcal{M}_{\mathcal{S}}(1-s,u,\tau)\, .
\end{equation}
Furthermore, a probabilistic bound for the stationary virtual delay  $\P (\W > w^\eps) \leq \eps$ of the system is obtained likewise by considering the limit of the function $\K$ as $t \to \infty$.
\end{theorem}
\noindent In what follows, we refer to the function $\K(s,\tau,t)$ above as the \textit{kernel}. 
The bound given by Theorem~\ref{thm:SOTABound} is applicable to any type of network as long as the Mellin transforms exist and are obtainable. Its usefulness is surmount when applied to the analysis of wireless fading channels as the Mellin transform $\mathcal{M}_{\S}(1-s,u,\tau)$ is already derived for many fading channel models in the literature, e.g., \cite{Schiessl:2015:FBL, Forssell:FBA, Zubaidy:TMM, Naghibi:wiretap},
which makes this approach an attractive one for wireless networks analysis. 

For the convenience of exposition of the results in this paper we define the function $V(s)$ as follows:
\begin{align}\label{eq:MellinV(s)}
V(s) \triangleq [\mathcal{M}_{\mathcal{S}}(1-s,u,\tau)]^{\frac{1}{\tau-u}} 
\end{align}

{\color{black} \textbf{Rayleigh block-fading channel} --
For Rayleigh block fading wireless channel that provides a service process with achievable Shannon capacity per slot, we have for channel $n$
\begin{equation}\label{eq:Rayleigh-Service-Process}
S_n(u,t) = W \sum_{i = u}^{t-1} \log_2 (1 + \gamma_{n}(i) ),
\end{equation}  
where $W$ is the bandwidth and $\gamma_{n}(i)$ is the instantaneous signal-to-noise ratio (SNR) for channel $n$ during time slot $i$. 
Since the channels are i.i.d. both across links and time slots, we write $\gamma_{n}(i)= \bar{\gamma}Y$, where $\bar{\gamma}$ is the average SNR and the channel gain $Y$ is an exponentially distributed random variable with unit mean. In this case, the Mellin transform of the service increment is given by:
\begin{align*}
\mathcal{M}_{\mathcal{S}}(1-s,u,\tau) &=\mathbb{E}[\S^{-s}(u,\tau)] = \mathbb{E}\left[\prod_{i=u}^{\tau-1}(1 + \gamma_{n}(i) )^{\frac{-sW}{\log 2}}\right]\nonumber\\
&= \prod_{i=u}^{\tau-1} \mathbb{E}\left[(1 + \gamma_{n}(i) )^{\frac{-sW}{\log 2}}\right]\nonumber\\
&= \prod_{i=u}^{\tau-1} \int_{0}^{\infty} (1 + \bar{\gamma}y_i )^{\frac{-sW}{\log 2}}e^{-y_i}dy_i \nonumber\\
&= \left[ e^{\frac{1}{\bar\gamma}} \bar{\gamma}^{\frac{-s W}{\log 2}} \Gamma\left(1-\frac{s W}{\log 2}, \bar \gamma^{-1}\right)\right]^{\tau-u}. 
\end{align*}
From~\eqref{eq:MellinV(s)}, we have
\begin{align}
V(s) = e^{\frac{1}{\bar\gamma}} \bar{\gamma}^{\frac{-s W}{\log 2}} \Gamma\left(1-\frac{s W}{\log 2}, \bar \gamma^{-1}\right) \,, 
\end{align}
}
where, $\Gamma(x,a)$ is the lower incomplete Gamma function.

\subsection{Stationary Bound for Transient Analysis}
For transient analysis, one may consider the existing stationary bound of Theorem~\ref{thm:SOTABound} and apply it to the equivalent model in Figure~\ref{fig:SysModel}.
However, this approach is limited as we assumed the arrival increments $\{a_i\}$ to be zero for $i > T$. 
In other words, the arrival process we consider is non-zero over a finite time horizon, while a (meaningfull) stationary bound can only be derived for an arrival process that has non-zero arrivals over infinite time horizon.
Nevertheless, it is worthwhile to establish some form of stationary bound as a reference for the more fine-grained transient analysis presented further below.
One straightforward way to establish such a reference is by assuming the deterministic instantaneous arrivals $\{a_i\}$ to occur over the infinite time horizon and invoke Theorem~\ref{thm:SOTABound}. 
In this subsection, we follow this approach and also discuss a first numerical evaluation.
 
In order to obtain a bound for the stationary virtual delay, one essentially has to determine the limit of the kernel, given in~\eqref{eq:function_M_Hussein}, as $t$ goes to infinity.
\begin{theorem}[Section V-C,~\cite{AlZubaidyTON}]\label{thm:stationary}
A probabilistic stationary end-to-end delay bound for a cascade of $N$ i.i.d. wireless channels with homogeneous average SNRs,  $(\sigma,\rho)$-bounded arrival traffic, and $(\sigma_c,\rho_c)$-bounded cross-traffic is given by
\begin{align*}
\P(\W > w) \leq &\inf_{s > 0} \Big\{\frac{e^{s(-\rho w+\sigma + N\sigma_c)}}{(1-V_0(s))^N}.\\ & \quad \quad \quad \min\{1,(V_0(s))^w(w+1)^{N-1}\}\Big \},
\end{align*}
where 
\begin{equation*}
V_0(s) = e^{s(\rho + \rho_c)}V(s)\, .
\end{equation*}
\end{theorem}
In order to apply this result to the model described above, we dimension the cross-traffic at each link using Eq.~\eqref{eq:InitBacklog} by setting $\sigma_c = x_\text{max}$ and $\rho_c = 0$.
In the rest of the paper, we refer to this reference bound given in Theorem~\ref{thm:stationary} as \textit{stationary bound}.

\begin{figure}[t]
\centering
\includegraphics[width = 3in]{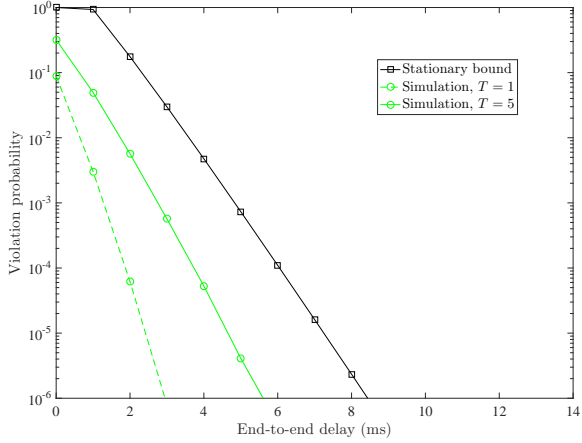}
\caption{Delay violation probability vs. end-to-end delay for a single link with SNR $= 5$ dB, $x_1 = 0$, $\rho = 20$, $\sigma = 0$.}
\label{fig:singleLink_x0_burst20_t1}
\end{figure}
To evaluate the viability of this approach, we compare in Figure~\ref{fig:singleLink_x0_burst20_t1} the transient violation probability obtained using simulations with the stationary bound for a single link wireless transmission.
The considered scenario is parameterized by $1$ ms slot duration, a burst arrival of size $20$ bits and no initial backlog, i.e., $\mathbf{x}=0$.
Accordingly, for determining the stationary bound we set the instantaneous arrival to a constant arrival process with rate $20$ bits per time slot.
In Figure~\ref{fig:singleLink_x0_burst20_t1}, we plot the resulting CCDF of the end-to-end delay from simulation and compare it to the stationary bound obtained using Theorem~\ref{thm:stationary} for two arrival processes with $T = 1$ and $T = 5$. 
The plot reveals the challenge with respect to transient analysis. 
We observe considerable gap between the simulations and the stationary bound.
For smaller $T$, the gap becomes arbitrarily large, and in particular the decay rate of the bound does not match the decay rate of the simulations.
This slackness is a direct result of the method used to obtain the stationary bound, since it is obtained by letting $t \rightarrow \infty$ in \eqref{eq:function_M_Hussein} which results in $\P(\W(t) > w) \rightarrow \P(\W > w)$ in Theorem \ref{thm:stationary}. This results in adding more terms to the summation which in turn increases the slackness in the bound. Although this delay bound proved to be useful for long-term stationary flows, Figure~\ref{fig:singleLink_x0_burst20_t1} shows that it is too loose for per-packet delay analysis.  
This motivates us to investigate a better bound for transient network performance, which we address in the next section.

\section{Transient Analysis} 
\label{sec:TransientBound}
In this section, we present our main contribution, namely, new  transient upper bounds for the delay violation probability for the $N$-hop wireless network described in Sections~\ref{sec:model}  and \ref{sec:Preliminaries}.
To this end, we present two different approaches and later (in Section \ref{sec:numerical}) compare the  probabilistic delay bounds resulting from both approaches. 
A first, straightforward approach is to start from the definition of the kernel as presented in Section~\ref{sec:Min_x_algebra} and derive a transient bound.
We refer to this as \textit{State-Of-The-Art-Transient (SOTAT)  bound}, as it is essentially an application of known results. 
However, as we will show later, this bound - while improving over the stationary bound - is still loose. This motivates us to present a new analysis for the transient behaviour of the system which results in a new \textit{Wireless Transient Bound} (WTB). Furthermore, we analyse the computational complexity of WTB and show how it can be used when the initial backlog information is delayed.
 
\subsection{State Of The Art Transient (SOTAT)  Bound}
In this subsection, we derive a transient upper bound for the violation probability using the results from~\cite{AlZubaidyTON}.
We first note that the bound on the violation probability of the virtual delay at time $t$ in Theorem~\ref{thm:SOTABound} corresponds to the transient analysis problem at hand. 
We hence focus on the components of the kernel. Since the sequence of arrivals is deterministic, we have
\begin{align}\label{eq:Arrivals}
\M_{\A} (s+1,u,t) = \mathbb{E}[(\A(u,t))^s] = [\A(t)/\A(u)]^s.
\end{align} 

Let $\S_0$ denote the dynamic SNR server \cite{Book:chang} that describes the network service offered by the multi-hop route to  the through traffic. The following lemma evaluates the Mellin transform of $\S_0$.
\begin{lemma}[Lemma 6,~\cite{AlZubaidyTON}]
Given $\sigma_c = x_\text{max}$ and $\rho_c = 0$, the Mellin transform of $\S_0(u,\tau)$ satisfies for  $s<1$ that
\begin{align}\label{eq:MellinTransformService}
\M_{\S_0} (1\!-\!s, u,\tau) \!\leq \! e^{sNx_\text{max}}\binom{N-1+\tau-u}{\tau-u}V(s)^{\tau-u}.
\end{align}
\end{lemma}
Substituting~\eqref{eq:Arrivals} and~\eqref{eq:MellinTransformService} into the definition of the kernel \eqref{eq:function_M_Hussein}, we obtain:
\begin{align}\label{eq:SOTABound}
\K(s,\tau,t) \leq & e^{sNx_\text{max} }\!\left[\sum_{u=0}^{t} \!\left(\!\frac{\A(t)}{\A(u)}\!\right)^s\!\binom{N\!-\!1\!+\!\tau\!-\!u}{\tau-u}V^{\tau-u}(s)\right].
\end{align}
The SOTAT bound is then computed by minimizing the RHS of~\eqref{eq:SOTABound} over $s$. For the case of a single link, a closed form expression can be obtained for the SOTAT bound under $(\sigma,\rho)$-bounded arrivals which is given by the following corollary.
\begin{corollary}
Assuming $\A(t)$ follows the $(\sigma,\rho)$-bounded traffic characterization, defined in~\eqref{eq:boundedarrival}, and for a single wireless link, an upper bound for $\P (\W(t) > w)$ is given by
\begin{align*}
\min_{s>0} \; \left\{ e^{s(x_1 - \rho w)}(V_0(s))^w\Big [e^{s \sigma}\cdot \frac{V_0(s) - (V_0(s))^{t+1}}{1-V_0(s)} + 1\Big] \right \}.
\end{align*}
\end{corollary}
\begin{proof}
Using~\eqref{eq:boundedarrival}, we obtain 
\begin{align}\label{eq:sigmarho}
[\A(t)/\A(u)]^{s} \leq e^{s(\sigma + \rho(t-u))}.
\end{align}
Now, using Theorem~\ref{thm:SOTABound} with $N=1$ and substituting~\eqref{eq:sigmarho} in~\eqref{eq:SOTABound}, we obtain
\begin{align*}
& \P (\W(t) > w)\\ 
& \leq \min_{s>0} \; \left\{ e^{sx_\text{max}}\Big[\sum_{u=0}^{t-1} e^{s (\sigma + \rho(t-u))}  V^{\tau-u}(s) + (V(s))^{\tau-t}\Big]\right \}\\ 
&= \min_{s>0} \; \left\{ e^{s(x_\text{max} - \rho w)}\Big[\sum_{u=0}^{t-1} e^{s \sigma}  V^{\tau-u}_0(s) + (V_0(s))^w\Big]\right \}\\
&= \min_{s>0} \; \left\{e^{s(x_1 - \rho w)}\Big [e^{s \sigma}  V_0^{\tau}(s) \frac{(V_0(s))^{-t} - 1}{({V_0(s)})^{-1} - 1} + (V_0(s))^w\Big]\right \} \\
&= \min_{s>0} \; \left\{e^{s(x_1 - \rho w)}(V_0(s))^w\Big [e^{s \sigma} \frac{V_0(s) - (V_0(s))^{t+1}}{1-V_0(s)} + 1\Big]\right \}.
\end{align*}
\end{proof}

Later, in Section~\ref{sec:numerical}, we will show that for multi-hop scenarios the SOTAT bound can become very loose in particular in cases with non-zero initial backlog. 
Nevertheless, it still captures the exponential decay rate of the delay tail distribution. 
The SOTAT bound slackness is mainly due to the fact that it is based on results that are initially derived for stationary settings. 
Although we believe that asymptotically we may not be able to improve on this bound, there is plenty of room for improvement for short sequences  (i.e., small $T$) and short delay target -- which is the case for most modern (and rapidly growing) MTC and CPS applications. 
To this end, we next derive the proposed WTB.



\begin{table*}[h]
\caption{Comparison of the different delay bounds for single wireless links with $(\sigma,\rho)$-bounded arrivals}
\centering
\begin{tabular}{| p{0.2\linewidth}|p{0.6\linewidth}|}
\hline
\\
Stationary & $
\inf_{s > 0} \left\{ e^{s(x_1 +\sigma -\rho w)} \cdot \frac{V^w_0(s)}{(1-V_0(s))} \cdot \max \{ V^{-w}_0(s), 1 \}\right \}
$\\
\hline \\
\shortstack[c]{SOTAT}  & $
\min_{s>0} \; \left\{ e^{s(x_1 +\sigma -\rho w)} \cdot \frac{V^w_0(s)}{(1-V_0(s))} \cdot \Big[V_0(s) - V^{t+1}_0(s) + \frac{1-V_0(s)}{e^{-s\sigma}} \Big] \right \}
$ \\
\hline \\
\shortstack[c]{Proposed WTB} & $
\min_{s>0}\left\{ e^{s(x_1 +\sigma -\rho w)} \cdot \frac{V^w_0(s)}{(1-V_0(s))} \cdot \Big[ V^{t}_0(s)  \cdot (1-V_0(s))  +   \frac{V_0(s)}{e^{s x_1}}  \cdot (1-V^{t-1}_0(s))\Big] \right\}
$\\
\hline
\end{tabular}
\label{tabel2}
\end{table*}

\subsection{Wireless Transient Bound (WTB)}
Our derivation of WTB is inspired by the bounding techniques used in~\cite{AlZubaidyTON}. However, we conduct independent analysis starting with the basic principles of network calculus and tailor the result, from the beginning, to our system with initial backlog. We note that our analysis is more involved due to the presence of the initial backlog $\mathbf{x}$.

We start by presenting our analysis for a single-hop scenario, i.e., for $N=1$. Then we generalize the obtained results for the multi-hop, $N>1$, case. In the following theorem, we state the proposed upper bound for the single hop case.
\begin{theorem}\label{thm:singleHop}
Given an SNR arrival sequence $\A$ traversing a wireless channel and given an initial backlog $x_1$, 
 an upper bound for the delay violation probability, $\P\{\W(t) > w\}$, is given by
\begin{align*}
\min_{s>0}\left[\A^s(t)e^{sx_1} V^{\tau}(s) + \sum_{u=1}^{t-1} [\A(t)/\A(u)]^s V^{\tau-u}(s)\right]\, ,
\end{align*}
where $\tau = t +w$.
\end{theorem}
\begin{proof}
Let $\S(t)$ be the SNR service process for wireless channel given by \eqref{eq:SNR-Service-Process}.
From the definition of dynamic server \cite{AlZubaidyTON}, we have for all $\tau \ge 0$ 
\begin{align*}
\D(\tau) \geq \underset{0\leq u \leq \tau}{\min} \{\S(\tau-u) \cdot \A(u) \cdot \A^c_1(u)\}. 
\end{align*}
Recall that $\A^c_1(0) = 1$ and $\A^c_1(u) = e^{x_1}$ for all $u > 0$. Now, the event $\{\W(t) > w\}$ is equivalent to the event that the cumulative departures at time $\tau$ is strictly less than the cumulative arrivals at time $t$ plus the initial backlog $x_1$. Therefore,
{\allowdisplaybreaks \begin{align}\label{eq:VioProb}
&\P\{\W(t) > w\} = \P\{\D(\tau) < \A(t) e^{x_1}\} \nonumber \\
&= \P\left\{\underset{0\leq u \leq \tau}{\min} [\S(\tau-u) \cdot \A(u) \cdot \A^c_1(u)] < \A(t) e^{x_1}\right\}\nonumber\\
&= \P\left\{ \!\{\S(\tau) \! < \! \A(t) e^{x_1}\}\! \cup \! \left(\bigcup_{u=1}^{\tau}\{\S(\tau-u) \cdot \A(u) \! < \! \A(t)\}\right) \!\right\} \nonumber\\
&= \P\left\{\! \{\S(\tau) \! < \! \A(t) e^{x_1}\} \! \cup \! \left(\bigcup_{u=1}^{t-1}\{\S(\tau-u) \cdot \A(u) \! < \! \A(t)\}\right)\! \right\} \nonumber\\
&\leq \P\{\S(\tau) \! < \! \A(t) e^{x_1}\} + \sum_{u=1}^{t-1}\P\{\S(\tau-u) \cdot \A(u) \! < \! \A(t)\} \nonumber \\
&\leq \min_{s>0}\Big[ [\A(t)]^se^{sx_1}\mathbb{E}[\S^{-s}(\tau)] \!+\! \sum_{u=1}^{t-1} \left[\frac{\A(t)}{\A(u)} \right]^s \! \mathbb{E}[\S^{-s}(\tau-u)]\Big] \,\nonumber\\
&= \min_{s>0}\left[[\A(t)]^se^{sx_1} V^{\tau}(s) + \sum_{u=1}^{t-1} \left[\frac{\A(t)}{\A(u)} \right]^s\!\! V^{\tau-u}(s)\right].
\end{align}}
In the third step above we have used the fact that $\P\{\S(\tau-u) \cdot \A(u) < \A(t)\} = 0$ for $u \geq t$. The fourth step utilizes the union bound and the fifth step follows from the moment bound.
\end{proof}

\begin{figure*}[!t]
\normalsize
\begin{align}\label{eq:NHopBound}
&\Phi(s) = V^{\tau}(s)\left[\binom{N +\tau - 2}{\tau - 1} \cdot \sum_{u=1}^{t-1} [\A(t)/\A(u)]^{s} V^{-u}(s) + \sum_{i=0}^{N-1} \binom{i +\tau - 1}{\tau - 1}[\A(t)]^s e^{s\sum_{n=1}^{N-i} x_n}\right].
\end{align}
\hrulefill
\vspace*{4pt}
\end{figure*}

Next, we extend the bound in Theorem~\ref{thm:singleHop} to the homogeneous $N$-hop case.  
Even though the analysis for the $N$-hop case is more involved, it essentially uses the same bounding techniques from the proof of Theorem~\ref{thm:singleHop}. 

\begin{theorem}\label{thm:NHop}
When the sequence $\A$ traverses the N-hop wireless network in Figure~\ref{fig:SysModel} and given the initial backlog vector $\mathbf x$, we compute  $\P\{\W(t) > w\} \leq \min_{s> 0}\; \Phi(s)$, where $\Phi(s)$ is given by~\eqref{eq:NHopBound}.
\end{theorem}
\begin{proof}
{\color{black} A key aspect we use in deriving the bound is to unfold the \mx-convolution starting with the arrival and service processes of link $N$ and then iteratively bound the departure processes $\{\D_n\}$ in the decreasing order of $n$. This allowed us to systematically account for the initial backlog at each node.} Full proof is given in  Appendix A.
\end{proof}
{\color{black} Observe that the expression in~\eqref{eq:NHopBound} is valid for any sequence of arrivals in the interval $[0,T]$. One may further simplify this expression by assuming that the cumulative arrival process obeys $(\sigma,\rho)$-bounded traffic characterization. This results in a simpler upper bound which is given by the following corollary.}
{\color{black}\begin{corollary}
Assuming $\A(t)$ obeys $(\sigma,\rho)$-bounded traffic characterization, defined in~\eqref{eq:boundedarrival}, the proposed transient bound in Theorem~\ref{thm:NHop} is reduced to the following:
\begin{align*}
\P\{\W(t) > w\} \leq & \min_{s> 0}\; V^{\tau}(s)\Big[ \sum_{i=0}^{N-1} \binom{i +\tau - 1}{\tau - 1} e^{s\sum_{n=1}^{N-i} x_n} \\
& + \binom{N +\tau - 2}{\tau - 1} \cdot e^{s(\sigma + \rho t)}\cdot \frac{1-V_0^{-t}(s)}{V_0(s)-1} \Big],
\end{align*}
where $V_0(s) = e^{s\rho}V(s)$.
\end{corollary}}
\begin{proof}
Using~\eqref{eq:sigmarho} in the summation part of the first term of $\Phi(s)$, given in~\eqref{eq:NHopBound}, we obtain
\begin{align*}
\sum_{u=1}^{t-1} [\A(t)/\A(u)]^{s} V^{-u}(s) &\leq \sum_{u=1}^{t-1} e^{s(\sigma + \rho (t-u))} V(s)^{-u}\\
&= e^{s(\sigma + \rho t)} \cdot \sum_{u=1}^{t-1} [e^{s\rho} V(s)]^{-u}\\
&= e^{s(\sigma + \rho t)}\cdot \frac{1-V_0^{-t}(s)}{V_0(s)-1}.
\end{align*}
The corollary follows by substituting the above inequality in~\eqref{eq:NHopBound}.
\end{proof}
\noindent For the case of a single link and for $(\sigma,\rho)$-bounded arrivals, we present different bounds in Tabel~\ref{tabel2}.

Next, we derive a probabilistic transient bound for the total backlog in the system. 

\textbf{Transient Backlog Bound} --
{\color{black} Given the initial backlog vector $\mathbf{x}$ and using Theorem~\ref{thm:NHop}, we derive an upper bound for the total backlog $B(t)$ in the system at time $t$ which 
we state  in the following corollary.}
{\color{black} \begin{corollary}\label{lem:backlog}
Given the initial backlog vector $\mathbf{x}$, an upper bound for the violation probability for the total backlog in the system at time $t$ is given by 
\begin{align*}
\P\{\B(t) > x\} \leq  \min_{s > 0} \quad e^{-xs}\Phi(s).
\end{align*}
\end{corollary}}
\begin{proof}
The proof is given in Appendix B.
\end{proof}

\subsection{Complexity Analysis}  
From~\eqref{eq:NHopBound}, we infer that the computational complexity for computing $\Phi(s)$ is $O(t+N)$ for $0\le t\le T$. Further, to obtain the WTB we need to solve the optimization problem of minimizing $\Phi(s)$ over $s > 0$. Thankfully, this is a convex optimization problem as $\Phi(s)$ is convex which we show next in the following lemma.
{\color{black}\begin{lemma}\label{lem:convexity}
Assuming $V(s)$ is differentiable, for $s > 0$, the function $\Phi(s)$ is convex.
\end{lemma}}
\begin{proof}
The proof is given in the Appendix C.
\end{proof}

{\color{black} We note that stationary bound givem in Theorem~\ref{thm:stationary} is known to be a convex optimization problem,~\cite{Petreska_16} and the objective function is in closed form. Therefore, WTB has a factor of $O(T+N)$ higher computational complexity compared with the stationary bound. This increase in computational complexity can be attributed to the fact that WTB does not restrict the sequence of arrivals and it carefully incorporates the initial backlog of each node. Finally, a similar analysis shows that the SOTAT bound has a factor of $O(T)$ higher computational complexity compared with the stationary bound.}
\begin{figure*}[!t]
\normalsize
\begin{align}\label{eq:NHopBound_delayedBacklog}
&\Phi_d(s) = V^{\tau}(s)\left[\binom{N +\tau - 2}{\tau - 1} \cdot \sum_{u=1}^{t-1} [\A'(t)/\A'(u)]^{s} V^{-u}(s) + \sum_{i=0}^{N-1} \binom{i +\tau - 1}{\tau - 1}[\A'(t)]^s e^{s\sum_{n=1}^{N-i} x_n(t_0-d)}\right].
\end{align}
\hrulefill
\vspace*{4pt}
\end{figure*}

{\color{black} \subsection{Delayed Backlog Information}} \label{subsec:backlog}
To compute WTB, node $1$ requires to know the backlogs of all nodes at time $t_0$. However, in practical systems, the current backlog information at node $n > 1$ may not be known to node $1$ instantaneously, instead it may be received with some delay due to the time it takes to rely this information back to node 1. For example, if a node relays its current backlog and the backlogs of the successor nodes to its predecessor node with a delay of one time slot, then the backlog at node $n > 1$ at time $t_0$ will be received by node $1$ at time $t_0+n-1$. In this case, at time $t_0$ node $1$ will have the backlog information of all the $N$ nodes at time $t_0-d$, where $ d = N - 1$. To incorporate this delayed backlog information in our transient analysis, we extend the delay bound for the case where at time $t_0$ node $1$ is only aware of the initial backlog vector at time $t_0-d$, which we refer to by $\mathbf{x}(t_0-d)$, for some $d > 0$.

To find the delay bound, we reuse the model in Figure~\ref{fig:SysDesc} and consequently apply Theorem~\ref{thm:NHop}. Recall that in the interval $[t_0-d,t_0-1]$ the arrivals are due to the stationary arrival process which we refer to as the overhead traffic $\A^\text{o}(t_0-d,t)$, for $t \in [t_0-d,t_0-1]$, in the following.
We define a cumulative arrival process $\A'(t)$ as follows.
\begin{eqnarray*}
\A'(t) = \left\{\begin{array}{lc}
	\! \A^\text{o}(t_0-d,t) \quad \quad \quad \quad \quad \quad \quad \quad \, t_0 - d \leq t \leq t_0,\\
    \! \A^\text{o}(t_0-d,t_0) \cdot \A(t_0,\min(t,t_0 + T + 1)) \, \,  t > t_0.
  \end{array}\right.
\end{eqnarray*}
Note that $\A'(t)$ only accounts for arrivals until time $t_0+T$, the time at which the time-critical sequence ends. 
Now, given $\mathbf{x}(t_0-d)$ the arrivals that occurred before $t_0-d$ can be ignored and we may consider that the system started at time $t_0-d$ with initial backlog $\mathbf{x}(t_0-d)$. This is equivalent to the model in Figure~\ref{fig:SysDesc}, except that the starting time is $t_0-d$ and the arrival sequence of our interest starts at $t_0$ instead of $t_0-d$. Note that our analysis that lead to the derivation of WTB is independent of the starting time, but depends on the cumulative arrival process since the starting time and the initial backlogs at the nodes. Thus, for the system that starts at time $t_0-d$, it is easy to see that Theorem~\ref{thm:NHop} can be applied using the cumulative arrival process $\A'(t)$ and $\mathbf{x}(t_0-d)$. This result is stated in the following theorem without proof.
\begin{theorem}\label{thm:delayedbacklog}
When the sequence $\A$ traverses the N-hop wireless network in Figure~\ref{fig:SysModel} and given the delayed initial backlog vector $\mathbf x(t_0-d)$, we compute  $\P\{\W(t) > w\} \leq \min_{s> 0}\; \Phi_d(s)$, where $\Phi_d(s)$ is given by~\eqref{eq:NHopBound_delayedBacklog}.
\end{theorem}
Note that the delay bound in Theorem~\ref{thm:delayedbacklog} is loose compared to the delay bound where the initial backlogs at time $t_0$ are known, i.e., $\mathbf{x}(t_0)$ is known to node $1$. The difference between these bounds depends on the information delay $d$ of the backlog measurements  and on the overhead traffic $\A^\text{o}$ intensity.

\section{Numerical Analysis}\label{sec:numerical}
In this section, we present a numerical evaluation of the proposed bounds and compare them to existing bounds and  simulation results. More specifically, we first present a comparison of the proposed WTB bound with the stationary bound and the SOTAT bound. 
We then evaluate the tightness of WTB in comparison with the violation probability obtained using simulation. 
{\color{black} Throughout the section, we assume Rayleigh bolck-fading channel model for the links and use the corresponding service process given in~\eqref{eq:Rayleigh-Service-Process}.} We also use a base set of parameters as follows: slot duration of $1$ ms, $\rho = 25, \sigma = 25$, bandwidth $W = 20$ kHz, average SNR $= 5$ dB, and initial backlog of $100$ bits, which is equally distributed along a multi-hop route (whenever  a multi-hop route is considered). 
Note that with an average SNR of $5$ dB, the average service rate (assuming Shannon capacity and a bandwidth of $20$ kHz) amounts to about $34$ bits per time slot.
Thus, by setting $\rho = 25$, the system basically operates in a transient regime where asymptotically it becomes stable.
We consider two types of arrival processes: (1) a burst arrival with $T=1, \sigma = 25, \rho=0$, modelling a single control packet, with $25$ bits, passed to the network at $t_0$; (2) an arrival process over multiple time slots with $T = 5, \sigma = 0, \rho = 25$, i.e. a train of packets that may represent a sequence of sensor data triggered through an event-based system.
The numerical bounds are computed using MATLAB and the Discrete Event Simulation is done using C.

\subsection{Comparison of Upper Bounds}
{\color{black} Recall that the stationary bound cannot be directly applied to the problem at hand. We use it as a reference by assuming the arrival process occurs over infinite time horizon and accordingly set the corresponding parameters in the bound.}
We start our numerical investigation by considering the burst arrival ($T=1$) model and a single wireless link with and without initial backlog.
The corresponding bounds and the simulation are presented in Figure~\ref{fig:singleLink_burst25_varx_5dB} where we plot the delay violation probability versus an increasing delay target $w$.
We observe that both WTB and SOTAT are significantly lower than the stationary bound. 
Note that the proposed WTB is not significantly lower than the SOTAT bound for such simple case of burst arrival. 
We also note that both WTB and SOTAT capture the tail decay rate of the delay distribution while the stationary bound is drastically off.
In Figure~\ref{fig:singleLink_rho25_x100_t5_5dB}, we consider the packet train arrival model with $T=5$ for the same single link system with an initial backlog of $100$, considering again the violation probability over an increasing delay target.
The figure reveals that in this case WTB outperforms the SOTAT bound over the entire range of delay target values by about one order of magnitude. 
In comparison to the simulated system performance, the proposed WTB is still about one order of magnitude higher, with a larger gap for longer delay targets.

Furthermore, in Figure~\ref{fig:20180519_TwoHop_rho25_x100_t5_varSNR}, we extend the scenario to a two-hop wireless system while considering the packet train arrival  with $T=5$.
The plot shows a comparison for two cases of average SNR: $5$ dB and $10$ dB. 
We observe that for the two-hop network the SOTAT bound performs worse than that for a one hop case. 
In particular for an average SNR of $10$ dB, i.e., at lower utilization ($43\%$), the proposed WTB bound is tighter by two orders of magnitude compared with the SOTAT bound. 
The relevance of this improvement is illustrated in Figure~\ref{fig:20180521_TwoHop_x100_rho25_T5_SNRvsw}.
{\color{black}Here we plot the predicted SNR requirement as computed using different bounds, for a QoS requirement of violation probability smaller than $10^{-9}$ for varying delay target $w$ in a two-hop network. 
We observe that for all delay target values, the proposed WTB bound results in a significantly lower average SNR requirement per link than the two comparison schemes. 
In other words, the proposed WTB bound provides a much better starting point, for instance, for accurate channel adaptation for mission-critical data transmissions in the short-term regime.
In summary, the results above demonstrate that the proposed WTB bound significantly outperforms the SOTAT and stationary bounds under general settings involving multiple hops, large initial backlogs and different utilizations. Therefore, we conclude that the SOTAT bound that is derived directly using the results from stationary analysis is inadequate for transient analysis and the additional computational complexity needed for WTB is well justified.   
\begin{figure}[t]
\centering
\includegraphics[width = 3in]{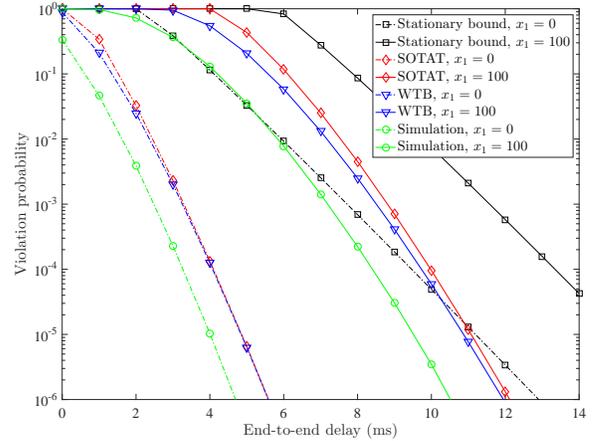}
\caption{Delay violation probability versus end-to-end delay for a single link with burst arrival ($T = 1$), SNR $= 5$ dB, $\rho = 0$, and $\sigma = 25$.}
\label{fig:singleLink_burst25_varx_5dB}
\end{figure}

\begin{figure}[t]
\centering
\includegraphics[width = 3in]{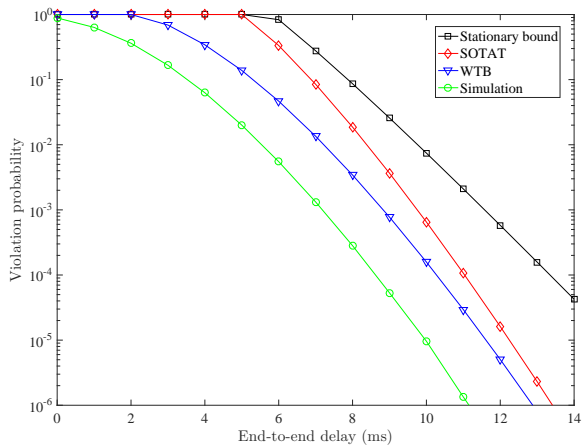}
\caption{Delay violation probability versus end-to-end delay for a single link with the packet train arrival process ($T = 5$), SNR $= 5$ dB, $x_1 = 100, \rho = 25$ and $\sigma = 0$.}
\label{fig:singleLink_rho25_x100_t5_5dB}
\end{figure}

\begin{figure}[t]
\centering
\includegraphics[width = 3in]{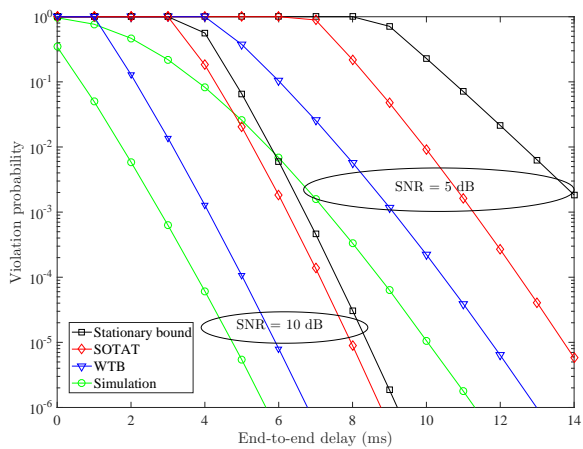}
\caption{Delay violation probability versus end-to-end delay for a two-hop network with the packet train arrival process ($T = 5$), $x_n = 100, \rho = 25$ and $\sigma = 0$.}
\label{fig:20180519_TwoHop_rho25_x100_t5_varSNR}
\end{figure}

\begin{figure}[t]
\centering
\includegraphics[width = 3in]{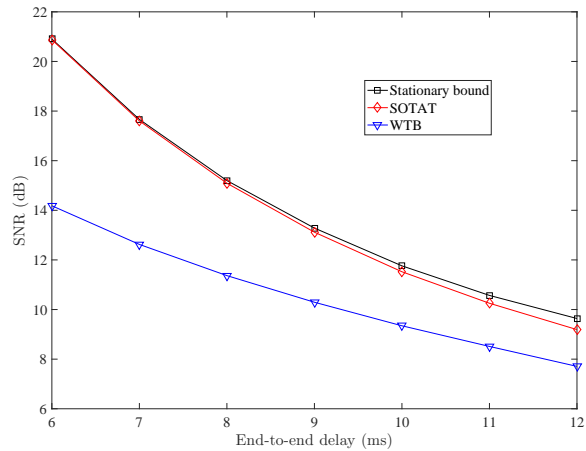}
\caption{SNR versus end-to-end delay for a delay violation probability requirement of $10^{-9}$ for a two-hop network with the packet train arrival process ($T = 5$), $x_n = 100, \rho = 25$ and $\sigma = 0$.}
\label{fig:20180521_TwoHop_x100_rho25_T5_SNRvsw}
\end{figure}


\subsection{Evaluation of WTB}
In the previous subsection, we have seen that WTB is typically within one order of magnitude of the simulated violation probability. 
In this subsection, we further investigate its performance for different parameter settings, and concentrate explicitly on comparing it with the simulated system behavior. 
In Figures~\ref{fig:TwoHop_rho25_x100_t5_varSNR} and~\ref{fig:TwoHop_rho25_Varx_t5_5dB}, we present performance results by varying the average SNR and the total initial backlog, respectively, in a two-hop network with the packet train arrival process. 
These results confirm that for a two-hop network, the gap between the proposed WTB bound and the simulated system performance remains around one order of magnitude despite significant variations in the average link SNRs or the initial backlog of the system.

In Figure~\ref{fig:ThreeHop_vart_x100_rho25_5dB}, we present the comparison for a three-hop network with burst arrival and train arrival processes. 
In this case, we observe that the gap increases significantly beyond one order of magnitude for the considered arrival process types. 
We expect this trend to continue as the number of hops increase. 
However, from all the figures above, we infer that WTB has a decay rate that always matches closely with the decay rate of the simulated violation probability. 
The significance of this property is that, an optimization of the proposed WTB bound can yield good heuristic solutions for the optimization of the end-to-end delay in the network operating in transient state.

\begin{figure}[t]
\centering
\includegraphics[width = 3in]{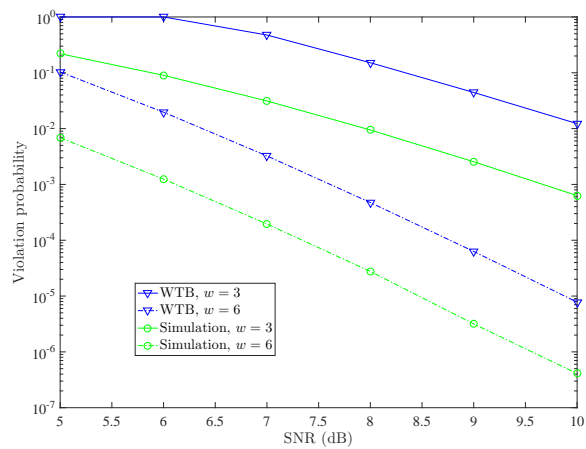}
\caption{Delay violation probability versus average SNR for a two-hop network for different target delays $w$, assuming the packet train arrival process $T = 5$, $x_n = 50, \rho = 25$, and $\sigma = 0$.}
\label{fig:TwoHop_rho25_x100_t5_varSNR}
\end{figure}

\begin{figure}[t]
\centering
\includegraphics[width = 3in]{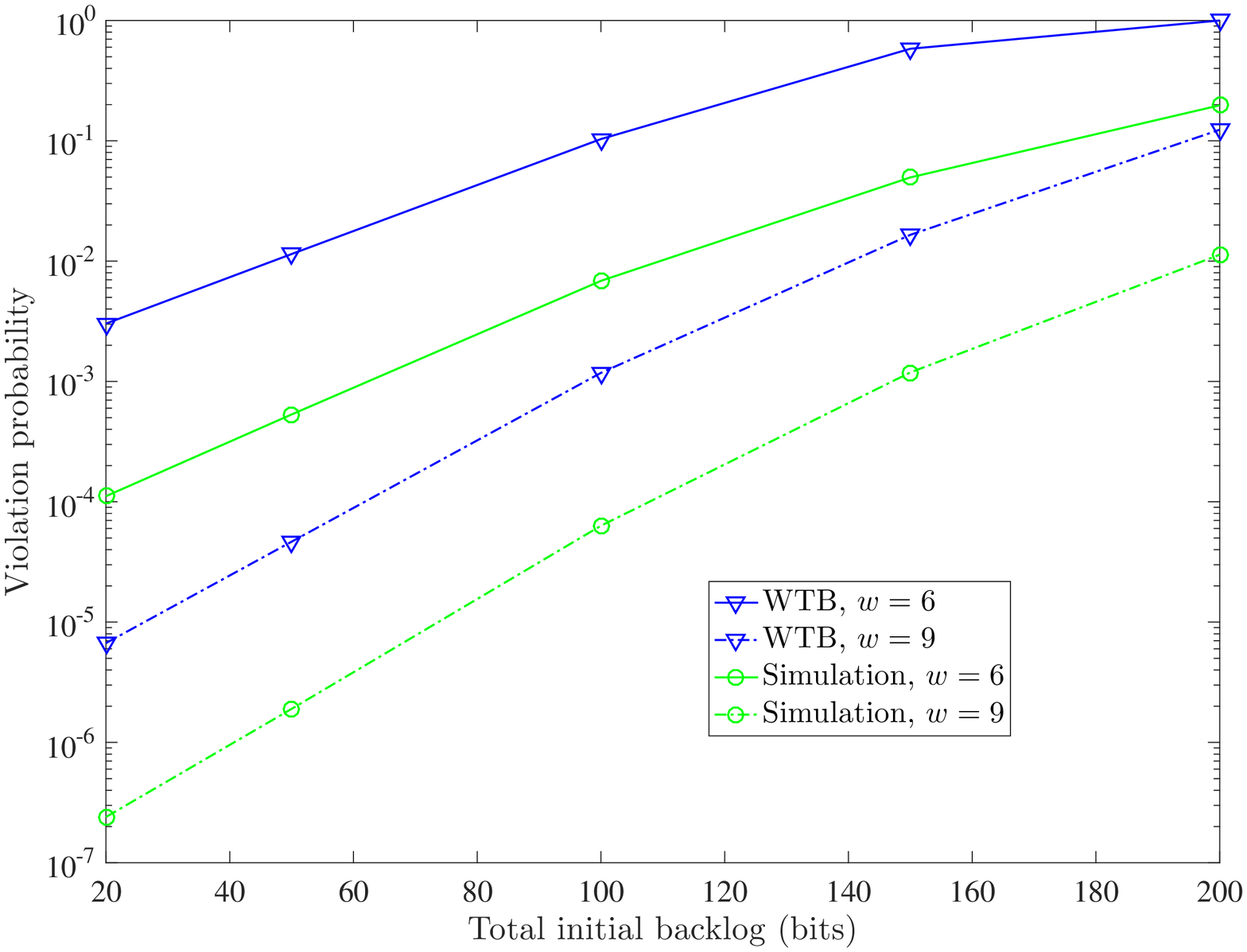}
\caption{Backlog violation probability versus total initial backlog for a two-hop network for different target delays $w$, assuming the packet train arrival process $T = 5$, $\rho = 25$, $\sigma = 0$, and SNR $= 5$ dB.}
\label{fig:TwoHop_rho25_Varx_t5_5dB}
\end{figure}

\begin{figure}[t]
\centering
\includegraphics[width = 3in]{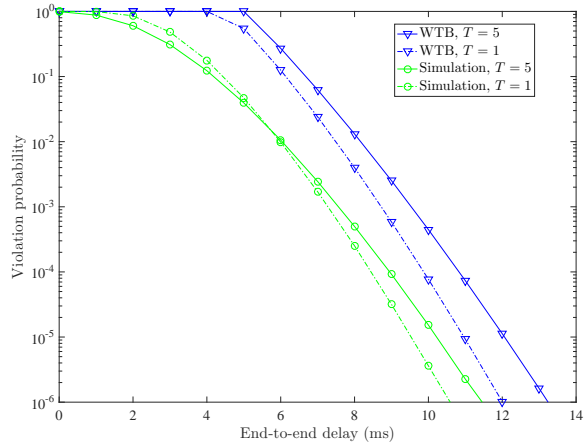}
\caption{Delay violation probability versus end-to-end delay for a three-hop network for different $T$, $x_n = 33, \rho = 25, \sigma = 0$, and SNR $= 5$ dB.}
\label{fig:ThreeHop_vart_x100_rho25_5dB}
\end{figure}

{\color{black}\subsubsection*{Delayed Backlog Information}
Finally, in Figures~\ref{fig:ThreeHop_DelayedBacklog_x100_rho25_5dB} and~\ref{fig:ThreeHop_DelayedBacklog_vard_x100_rho25_5dB} we study the impact of delayed backlog information on the predicted system behavior as captured by the bound presented in Section \ref{subsec:backlog}. 
In these figures we use the packet train arrival model while varying on the x-axis the target delay $w$ as well as backlog information delay $d$. 
From Figure \ref{fig:ThreeHop_DelayedBacklog_x100_rho25_5dB}, we observe similar trends as before for varying $w$, but recognize instantly a cost of the outdated backlog information of between a factor of 2 up to 5.
In Figure~\ref{fig:ThreeHop_DelayedBacklog_vard_x100_rho25_5dB} we observe that - in correspondence to the previous figure - the WTB bound becomes loose as $d$ increases by the same ratios as observed previously. 
This is expected and is a consequence of using union bound, in which the number of terms increase as $d$ increases.}
\begin{figure}[t]
\centering
\includegraphics[width = 3in]{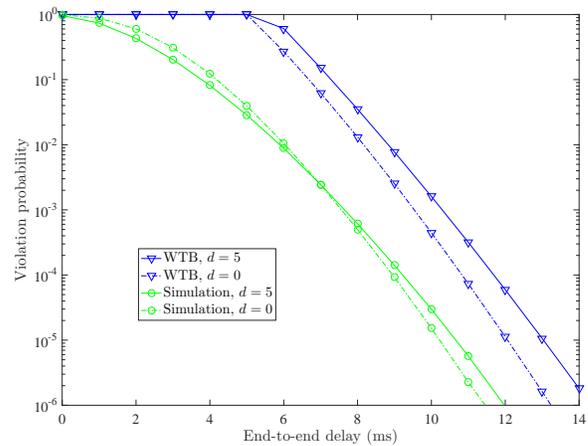}
\caption{Delay violation probability versus end-to-end delay for a three-hop network for different $d$, assuming the packet train arrival process with $T = 5, x_n = 33, \rho = 25, \sigma = 0$, and SNR $= 5$ dB.}
\label{fig:ThreeHop_DelayedBacklog_x100_rho25_5dB}
\end{figure}

\begin{figure}[t]
\centering
\includegraphics[width = 3in]{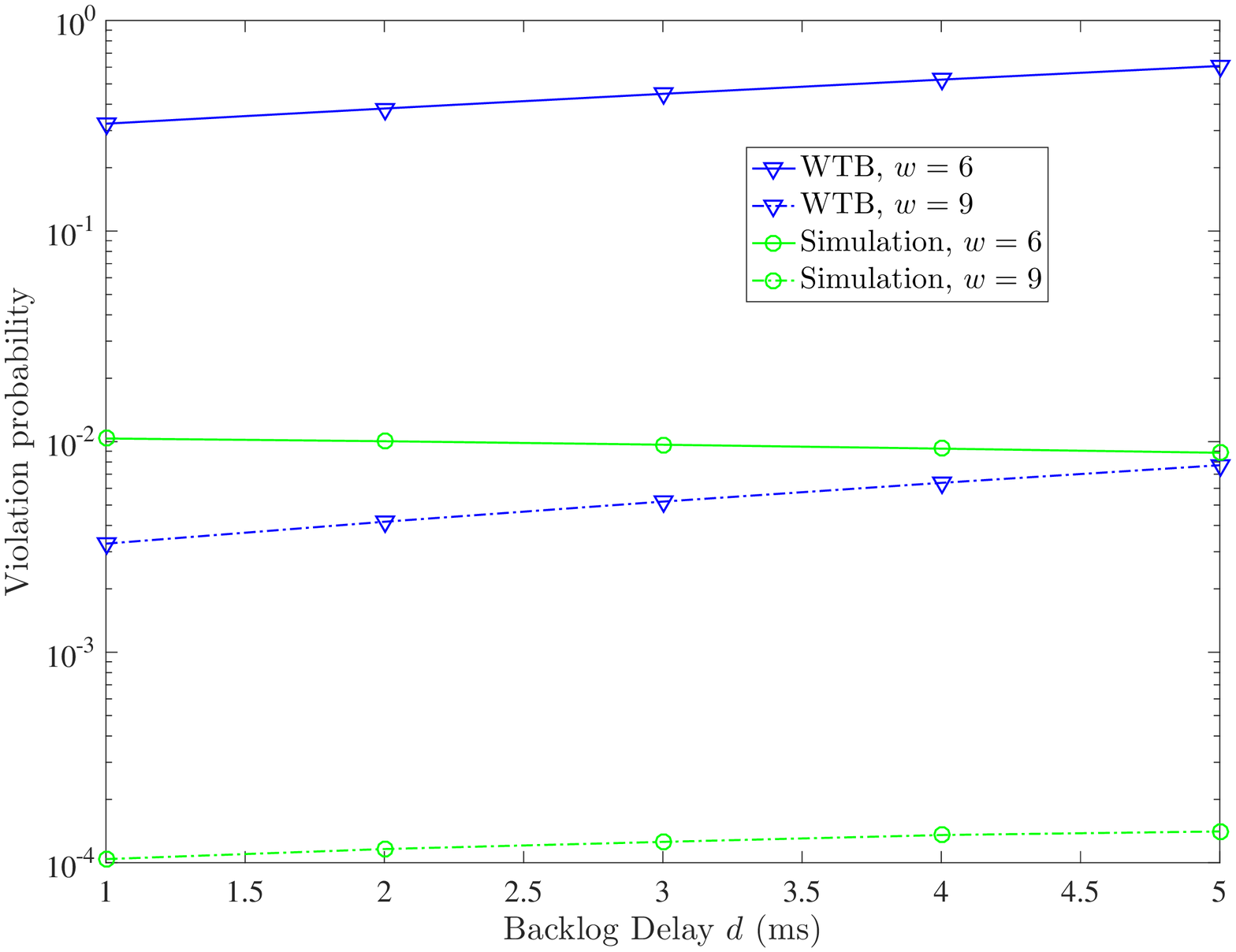}
\caption{Delay violation probability versus backlog information delay (d) for a three-hop network for different $w$, assuming the packet train arrival process with $T = 5, x_n = 33, \rho = 25, \sigma = 0$, and SNR $= 5$ dB.}
\label{fig:ThreeHop_DelayedBacklog_vard_x100_rho25_5dB}
\end{figure}

\section{Conclusions}
\label{sec:conclusions}
We have studied the problem of characterizing the end-to-end delay of a sequence of time-critical control packets traversing through a multi-hop wireless network with non-zero initial backlog at each hop. As this requires the network to be analysed in the transient state, we attempt to find upper bounds for the end-to-end delay using stochastic network calculus. We have studied the state-of-the-art upper bounds and have demonstrated their poor performance for the problem at hand. We have derived WTB by using the first principles of network calculus and the state-of-the-art bounding techniques. {\color{black} A key aspect of WTB is that it carefully incorporates the known initial backlog in the network. We also extended WTB for the case where the initial backlog information is delayed.} Through extensive simulations we have showed that WTB is significantly better than the alternatives. Also, we have observed that its decay rate closely matches the decay rate of the simulated violation probability. {\color{black} We believe that this key feature of WTB makes it a useful metric in the design and optimization of the networks for safety-critical machine-type applications.}

\bibliographystyle{IEEEtran}
\bibliography{Transient-Analysis-v2}

\begin{thebibliography}{10}
\providecommand{\url}[1]{#1}
\csname url@samestyle\endcsname
\providecommand{\newblock}{\relax}
\providecommand{\bibinfo}[2]{#2}
\providecommand{\BIBentrySTDinterwordspacing}{\spaceskip=0pt\relax}
\providecommand{\BIBentryALTinterwordstretchfactor}{4}
\providecommand{\BIBentryALTinterwordspacing}{\spaceskip=\fontdimen2\font plus
\BIBentryALTinterwordstretchfactor\fontdimen3\font minus
  \fontdimen4\font\relax}
\providecommand{\BIBforeignlanguage}[2]{{%
\expandafter\ifx\csname l@#1\endcsname\relax
\typeout{** WARNING: IEEEtran.bst: No hyphenation pattern has been}%
\typeout{** loaded for the language `#1'. Using the pattern for}%
\typeout{** the default language instead.}%
\else
\language=\csname l@#1\endcsname
\fi
#2}}
\providecommand{\BIBdecl}{\relax}
\BIBdecl

\bibitem{Yilmaz-EricssonResearch}
O.~N.~C. Yilmaz, Y.~P.~E. Wang, N.~A. Johansson, N.~Brahmi, S.~A. Ashraf, and
  J.~Sachs, ``Analysis of ultra-reliable and low-latency 5g communication for a
  factory automation use case,'' in \emph{Proc. IEEE International Conference
  on Communication Workshop (ICCW)}, Jun. 2015, pp. 1190--1195.

\bibitem{Dahlman-EricssonResearch}
E.~Dahlman, G.~Mildh, S.~Parkvall, J.~Peisa, J.~Sachs, Y.~Selén, and J.~Sköld,
  ``{5G} wireless access: requirements and realization,'' \emph{IEEE
  Communications Magazine}, vol.~52, no.~12, pp. 42--47, Dec. 2014.

\bibitem{SafetyRequirements}
J.~Hedberg, ``Safety requirements specifications guideline,'' \emph{SP Swedish
  National Testing and Research Institute}, 2005.

\bibitem{Book:morse}
P.~Morse, \emph{{Queues, Inventories and Maintenance: The Analysis of
  Operational Systems with Variable Demand and Supply}}.\hskip 1em plus 0.5em
  minus 0.4em\relax Wiley, 1958.

\bibitem{Zhang_91}
J.~Zhang and E.~J. Coyle, ``The transient solution of time-dependent {M/M/1}
  queues,'' \emph{IEEE Transactions on Information Theory}, vol.~37, no.~6, pp.
  1690--1696, Nov. 1991.

\bibitem{Matis_01}
T.~Matis and R.~Feldman, ``{Transient Analysis of State-Dependent Queuing
  Networks via Cumulative Functions},'' \emph{Journal of Applied Probability},
  vol.~38, 2001.

\bibitem{Czachorski_15}
T.~Czachórski, \emph{Queueing Models for Performance Evaluation of Computer
  Networks - Transient State Analysis}.\hskip 1em plus 0.5em minus 0.4em\relax
  Springer International Publishing Switzerland, 2015.

\bibitem{Mellia_02}
M.~Mellia and H.~Zhang, ``{TCP} model for short lived flows,'' \emph{IEEE
  Communications Letters}, vol.~6, no.~2, pp. 85--87, Feb. 2002.

\bibitem{wang_96}
C.~Wang, D.~Logothetis, K.~Trivedi, and I.~Viniotis, ``Transient behavior of
  {ATM} networks under overloads,'' in \emph{Proc. IEEE Infocom}, Mar. 1996.

\bibitem{wu}
W.~Dapeng and R.~Negi, ``Effective capacity: A wireless link model for support
  of quality of service,'' \emph{IEEE Transactions Wireless Communications},
  vol.~2, no.~4, pp. 630--643, Jul. 2003.

\bibitem{jiang}
Y.~Jiang, ``A basic stochastic network calculus,'' \emph{SIGCOMM Comput.
  Commun. Rev.}, vol.~36, no.~4, pp. 123--134, Aug. 2006.

\bibitem{fidler_mgf}
M.~Fidler, ``An end-to-end probabilistic network calculus with moment
  generating functions,'' in \emph{Proc. IEEE International Workshop on Quality
  of Service (IWQoS)}, Jun. 2006, pp. 261--270.

\bibitem{AlZubaidyTON}
H.~Al-Zubaidy, J.~Liebeherr, and A.~Burchard, ``Network-layer performance
  analysis of multihop fading channels,'' \emph{IEEE/ACM Transactions on
  Networking}, vol.~24, no.~1, pp. 204--217, Feb. 2016.

\bibitem{alzubaidyInfocom13}
------, ``A (min, $\times$) network calculus for multi-hop fading channels,''
  in \emph{Proc. IEEE Infocom}, Apr. 2013, pp. 1833--1841.

\bibitem{itc14}
N.~Petreska, H.~Al-Zubaidy, and J.~Gross, ``Power minimization for industrial
  wireless networks under statistical delay constraints,'' in \emph{Proc.
  International Teletraffic Congress (ITC)}, Sep. 2014, pp. 1--9.

\bibitem{icc15}
N.~Petreska, H.~Al-Zubaidy, R.~Knorr, and J.~Gross, ``On the recursive nature
  of end-to-end delay bound for heterogeneous wireless networks,'' in
  \emph{Proc. IEEE International Conference on Communications 2015 (ICC)}, June
  2015, pp. 5998--6004.

\bibitem{jiang_markov}
K.~Zheng, F.~Liu, L.~Lei, C.~Lin, and Y.~Jiang, ``Stochastic performance
  analysis of a wireless finite-state markov channel,'' \emph{IEEE Transactions
  on Wireless Communications}, vol.~12, no.~2, pp. 782--793, Feb. 2013.

\bibitem{florin_infocom}
F.~Ciucu, R.~Khalili, Y.~Jiang, L.~Yang, and Y.~Cui, ``Towards a system
  theoretic approach to wireless network capacity in finite time and space,''
  in \emph{Proc. IEEE INFOCOM}, Apr. 2014, pp. 2391--2399.

\bibitem{florin_capacity}
F.~Ciucu, ``Non-asymptotic capacity and delay analysis of mobile wireless
  networks,'' \emph{SIGMETRICS Perform. Eval. Rev.}, vol.~39, no.~1, pp.
  359--360, Jun. 2011.

\bibitem{fidler_globecom}
M.~Fidler, ``A network calculus approach to probabilistic quality of service
  analysis of fading channels,'' in \emph{Proc. IEEE GLOBECOM}, Nov. 2006, pp.
  1--6.

\bibitem{Becker-Fidler-ITC2015}
N.~Becker and M.~Fidler, ``A non-stationary service curve model for performance
  analysis of transient phases,'' in \emph{Proc. International Teletraffic
  Congress (ITC)}, Sep. 2015, pp. 116--124.

\bibitem{Cruz1991}
R.~L. Cruz, ``A calculus for network delay, part {I}: Network elements in
  isolation,'' \emph{IEEE Transactions on Information Theory}, vol.~37, no.~1,
  pp. 114--131, Jan 1991.

\bibitem{Book:chang}
C.-S. Chang, \emph{Performance Guarantees in Communication Networks}.\hskip 1em
  plus 0.5em minus 0.4em\relax Springer-Verlag, 2000.

\bibitem{Schiessl:2015:FBL}
S.~Schiessl, J.~Gross, and H.~Al-Zubaidy, ``Delay analysis for wireless fading
  channels with finite blocklength channel coding,'' in \emph{Proc. ACM MSWiM},
  2015, pp. 13--22.

\bibitem{Forssell:FBA}
H.~Forssell, R.~Thobaben, H.~Al-Zubaidy, and J.~Gross, ``On the impact of
  feature-based physical layer authentication on network delay performance,''
  in \emph{Proc. IEEE GLOBECOM}, Dec. 2017, pp. 1--6.

\bibitem{Zubaidy:TMM}
H.~Al-Zubaidy, V.~Fodor, G.~Dán, and M.~Flierl, ``Reliable video streaming with
  strict playout deadline in multihop wireless networks,'' \emph{IEEE
  Transactions on Multimedia}, vol.~19, no.~10, pp. 2238--2251, Oct. 2017.

\bibitem{Naghibi:wiretap}
F.~Naghibi, S.~Schiessl, H.~Al-Zubaidy, and J.~Gross, ``Performance of wiretap
  rayleigh fading channels under statistical delay constraints,'' in
  \emph{Proc. IEEE ICC}, May 2017, pp. 1--7.

\bibitem{Petreska_16}
N.~Petreska, H.~Al{-}Zubaidy, R.~Knorr, and J.~Gross, ``Power-minimization
  under statistical delay constraints for multi-hop wireless industrial
  networks,'' \emph{CoRR}, vol. abs/1608.02191, 2016.

\end{thebibliography}

\appendix

\subsection{Proof of Theorem~\ref{thm:NHop}}

\begin{figure*}[!t]
\normalsize
\setcounter{equation}{20}
\begin{align}\label{eq:NHop2}
&\P\{\W(t) > w\}  
\leq \sum_{i=0}^{N-1}\sum_{u=1}^{\tau}\sum_{u_1=1}^{u}\ldots\sum_{u_{i-1}=1}^{u_{i-2}} \P\Big \{\S_{N-i} (u_{i-1})\cdot \Pi_{n=1}^{i-1} \S_{N-n}(u_{n-1}-u_{n}) \cdot \S_N(\tau-u) < \A(t) \cdot e^{\sum_{n=1}^{N-i} x_n}\Big \} \nonumber\\ & \quad \quad \quad \quad \quad \quad \quad \quad + \sum_{u=1}^{\tau}\sum_{u_1=1}^{u}\ldots\sum_{u_{N-1}=1}^{u_{N-2}} \P\Big \{ \A(u_{N-1})\cdot\Pi_{n=1}^{N-1} \S_{N-n}(u_{n-1}-u_{n})\cdot \S_N(\tau-u) < \A(t)\Big \}.
\end{align}
\hrulefill
\vspace*{4pt}
\end{figure*}

\begin{figure*}[!t]
\normalsize
\begin{align}\label{eq1:NHop1_firstterm}
& \sum_{i=0}^{N-1}\sum_{u=1}^{\tau}\sum_{u_1=1}^{u}\ldots\sum_{u_{i-1}=1}^{u_{i-2}} \P\{\S_{N-i} (u_{i-1})\cdot\Pi_{n=1}^{i-1} \S_{N-n}(u_{n-1}-u_{n})\cdot \S_N(\tau-u) < \A(t) \cdot e^{\sum_{n=1}^{N-i} x_n}\}\nonumber \\
& \leq \sum_{i=0}^{N-1}\sum_{u=1}^{\tau}\sum_{u_1=1}^{u}\ldots\sum_{u_{i-1}=1}^{u_{i-2}}\min_{s > 0} \; [\A(t)]^s \cdot e^{s(\sum_{n=1}^{N-i} x_n)} \mathbb{E}[\{\S_{N-i} (u_{i-1})\cdot\Pi_{n=1}^{i-1} \S_{N-n}(u_{n-1}-u_{n})\cdot \S_N(\tau-u)\}^{-s}]\nonumber \\
& \leq \min_{s > 0}\; \sum_{i=0}^{N-1} [\A(t)]^s \cdot e^{s\sum_{n=1}^{N-i} x_n}V^{\tau}(s)\sum_{u=1}^{\tau}\sum_{u_1=1}^{u}\ldots\sum_{u_{i-1}=1}^{u_{i-2}} 1 = \min_{s > 0}\; \sum_{i=0}^{N-1} \binom{i +\tau - 1}{\tau - 1} [\A(t)]^s \cdot e^{s\sum_{n=1}^{N-i} x_n}V^{\tau}(s).
\end{align}
\hrulefill
\vspace*{4pt}
\end{figure*}

\begin{figure*}[!t]
\normalsize
\begin{align}\label{eq1:NHop1_secondterm}
&\sum_{u=1}^{\tau}\sum_{u_1=1}^{u}\ldots\sum_{u_{N-1}=1}^{u_{N-2}} \P\{ \A(u_{N-1})\cdot  \Pi_{n=1}^{N-1} \S_{N-n}(u_{n-1}-u_{n})\cdot \S_N(\tau-u) < \A(t)\} \nonumber \\
&\leq \sum_{u=1}^{\tau}\sum_{u_1=1}^{u}\ldots\sum_{u_{N-1}=1}^{t-1} \P\{ \A(u_{N-1})\cdot  \Pi_{n=1}^{N-1} \S_{N-n}(u_{n-1}-u_{n})\cdot \S_N(\tau-u) < \A(t)\} \nonumber \\
& = \sum_{u_{N-1}=1}^{t-1} \P\{ \A(u_{N-1})\cdot  \Pi_{n=1}^{N-1} \S_{N-n}(u_{n-1}-u_{n})\cdot \S_N(\tau-u) < \A(t)\} \sum_{u=1}^{\tau}\sum_{u_1=1}^{u}\ldots\sum_{u_{N-2}=1}^{u_{N-3}} 1 \nonumber \\
& \leq \binom{N +\tau - 2}{\tau - 1} \sum_{u_{N-1}=1}^{t-1} \min_{s > 0}\; [\A(t)/\A(u_{N-1})]^{s}\mathbb{E}[\{\Pi_{n=1}^{N-1} \S_{N-n}(u_{n-1}-u_{n}) \cdot \S_N(\tau-u))\}^{-s}] \nonumber \\
& \leq \binom{N +\tau - 2}{\tau - 1} \cdot \min_{s > 0}\; \sum_{u=1}^{t-1} [\A(t)/\A(u)]^{s} V^{\tau-u}(s).
\end{align}
\hrulefill
\vspace*{4pt}
\end{figure*}

Recall that  $\D_{N}(t) \geq \A_{N}\otimes \S_{N}(t)$ and $\A_{N}(t) = \D_{N-1}(t) \cdot \A^c_N(t)$. Since the event $\{\W(t) > w\}$ is equivalent to the event that the cumulative departures at node $N$ at time $\tau$ is strictly less than the cumulative arrivals by time $t$ plus the total initial backlog $\sum_{n=1}^N x_n$, we have
{\allowdisplaybreaks
\begin{equation}\label{eq:VioProb_NHop1}
\begin{aligned}
& \P\{\W(t) > w\} = \P\{\D_N(\tau) < \A(t) \cdot e^{ \sum_{n=1}^N x_n }\} \\
& \leq \P\{\A_N \otimes \S_N(\tau)  < \A(t) \cdot e^{\sum_{n=1}^N x_n}\} \\
& = \P\{(\D_{N-1} \cdot \A^c_N)\otimes \S_N(\tau)  < \A(t) \cdot e^{\sum_{n=1}^N x_n}\} \\
& =  \P\Big\{\underset{0\leq u \leq \tau}{\min} [\D_{N-1}(u) \cdot  \A^c_N(u) \\
& \quad \quad \quad \quad \quad \quad\cdot \S_N(\tau-u)] < \A(t) \cdot e^{\sum_{n=1}^N x_n}\Big\}\\
& =  \P\Big\{\{\S_N(\tau) < \A(t) \cdot e^{\sum_{n=1}^N x_n}\} \\ &\quad\quad \cup (\bigcup_{u=1}^{\tau}\{\D_{N-1}(u)\cdot\S_N(\tau-u) < \A(t) \cdot e^{\sum_{n=1}^{N-1} x_n}\})\Big\}  \\
& \leq  \P\{\S_N(\tau) < \A(t) \cdot e^{\sum_{n=1}^N x_n}\} \\ &\quad\quad + \sum_{u=1}^{\tau}\P\{\D_{N-1}(u)\cdot \S_N(\tau-u) < \A(t) \cdot e^{\sum_{n=1}^{N-1} x_n}\}\}.
\end{aligned}
\end{equation}
}
In the following we find an upper bound for the probabilities in the summation term of the last step in~\eqref{eq:VioProb_NHop1}. 
{\allowdisplaybreaks
\begin{align}\label{eq:NHop1_secondterm}
&\P\{\D_{N-1}(u)\cdot \S_N(\tau-u) < \A(t) \cdot e^{\sum_{n=1}^{N-1} x_n}\} \nonumber\\
&\leq \medmath{ \P\{(\D_{N-2} \cdot \A^c_{N-1})\otimes \S_{N-1}(u)\cdot \S_N(\tau-u)  < \A(t) \cdot e^{\sum_{n=1}^{N-1} x_n}\} } \nonumber\\
& =  \P\Big\{\underset{0\leq u_1 \leq u}{\min} [\D_{N-2}(u_1) \cdot  \A^c_{N-1}(u_1) \nonumber  \\ 
& \quad\quad\quad \cdot \S_{N-1}(u-u_1) \cdot \S_N(\tau-u)] < \A(t) \cdot e^{\sum_{n=1}^{N-1} x_n}\Big\}\nonumber\\
& \leq \P\{\S_{N-1}(u) \cdot \S_N(\tau-u) < \A(t) \cdot e^{ \sum_{n=1}^{N-1} x_n}\} + \nonumber\\ &   \medmath{ \sum_{u_1=1}^{u} \P\{\D_{N-2}(u_1)\cdot \S_{N-1}(u-u_1)\cdot \S_N(\tau-u) < \A(t) \cdot e^{\sum_{n=1}^{N-2}x_n}\} }.
\end{align}
}
Substituting~\eqref{eq:NHop1_secondterm} in~\eqref{eq:VioProb_NHop1}, we obtain
\begin{align*}
& \P\{\W(t) > w\} \nonumber \\ 
&\leq \P\{\S_N(\tau) < \A(t) \cdot e^{\sum_{n=1}^N x_n}\} + \nonumber \\  & \quad  \P\{\S_{N-1}(u) \cdot \S_N(\tau-u) < \A(t)\cdot e^{ \sum_{n=1}^{N-1} x_n}\} + \nonumber \\
& \quad  \medmath{ \sum_{u=1}^{\tau} \sum_{u_1=1}^{u}\hspace{-.15cm} \P\Big \{\D_{N-2}(u_1) \S_{N-1}(u-u_1) \S_N(\tau-u)\hspace{-.1cm} < \hspace{-.1cm} \A(t) \cdot e^{\sum_{n=1}^{N-2}x_n}\Big \} }.
\end{align*}

One can again use similar manipulation as in~\eqref{eq:NHop1_secondterm} to bound the probabilities in the double summation of the RHS of the above inequality. 
Repeating this step iteratively, and using the convention $u_0 = u$ and $u_{-1} = \tau$, we arrive at~\eqref{eq:NHop2}. The first and second terms in the RHS of~\eqref{eq:NHop2} are upper bounded as shown in~\eqref{eq1:NHop1_firstterm} and~\eqref{eq1:NHop1_secondterm}, respectively. 
{\color{black} In the first inequality of~\eqref{eq1:NHop1_firstterm} we have used the moment bound and in the second inequality we have used the fact}
\begin{align*}
\sum_{u=1}^{\tau}\sum_{u_1=1}^{u}\ldots\sum_{u_{i-1}=1}^{u_{i-2}} 1 = \binom{i +\tau - 1}{\tau - 1}.
\end{align*}
{\color{black} In the first inequality of~\eqref{eq1:NHop1_secondterm}, we have used the fact that the probability terms are zero for $u_{N-1} \geq t$.}

Finally, substituting~\eqref{eq1:NHop1_firstterm} and~\eqref{eq1:NHop1_secondterm} in~\eqref{eq:NHop2}, we obtain the result.

{\color{black}\subsection{Proof of Corollary~\ref{lem:backlog}}}

\begin{align}\label{eq:backlog_derivation}
&\P\{\B(t) > e^x\}\nonumber = \P\{\A(t) \cdot e^{\sum_{n=1}^N x_n} / \D(t) > e^x \} \nonumber \\
&=\P\{\D(t) < \A(t)\cdot e^{\sum_{n=1}^N x_n - x} \} 
\end{align}

Note that the expression in~\eqref{eq:backlog_derivation} is same as the expression in the first step of~\eqref{eq:VioProb_NHop1}, except for the additional factor $e^{-x}$. Therefore, we repeat the same steps of the proof of Theorem~\ref{thm:NHop} and arrive at the desired result. 

{\color{black} \subsection{Proof of Lemma~\ref{lem:convexity}}}
Since sum of convex functions is convex, it is sufficient to prove that the terms $[\A(t)/\A(u)]^{s} V^{\tau-u}(s)$ and $V^{\tau}e^{s\sum_{n=1}^{N-i} x_n}$ are convex. We will show that $[\A(t)/\A(u)]^{s} V^{\tau-u}(s)$ is convex and proof for the latter term is similar. In the following we first show that $V^{\tau-u}(s)$ is convex. From~\eqref{eq:MellinV(s)}, we have $V^{\tau-u}(s) = \mathbb{E}[\S^{-s}(u,\tau)]$. Therefore, for any two positive real numbers $s_1$ and $s_2$, and $0 \leq \theta \leq 1$, we have
{\allowdisplaybreaks
\begin{align*}
& V^{\tau-u}(\theta s_1 + (1-\theta) s_2) = \mathbb{E}[\S^{-(\theta s_1 + (1-\theta) s_2)}(u,\tau)]\\
&= \mathbb{E}[e^{-(\theta s_1 + (1-\theta) s_2) S(u,\tau)}]\\
&\leq \mathbb{E}[\theta e^{-s_1 S(u,\tau)} + (1-\theta)e^{-s_2 S(u,\tau)}]\\
&= \theta \mathbb{E}[\S^{-s_1}(u,\tau)] + (1-\theta)\mathbb{E}[\S^{-s_2}(u,\tau)]\\
&= \theta \mathbb{E}[\S^{-s_1}(u,\tau)] + (1-\theta)\mathbb{E}[\S^{-s_2}(u,\tau)]\\
&= \theta V^{\tau-u}(s_1) + (1-\theta)V^{\tau-u}(s_2).
\end{align*}
}
In the third step above we have used the fact that $e^{-sS(u,\tau)}$ is convex in $s$, for $s > 0$. Therefore, $V^{\tau-u}(s)$ is convex. It is easy to see that $[\A(t)/\A(u)]^{s}$ is convex by showing it second derivative is positive. 

Now, we show that the second derivative of $[\A(t)/\A(u)]^{s} V^{\tau-u}(s)$ is non-negative. Using chain rule iteratively, we obtain
\begin{align*}
& \frac{d^2 ([\A(t)/\A(u)]^{s} V^{\tau-u}(s))}{d s^2} = [\A(t)/\A(u)]^{s}\frac{d^2 (V^{\tau-u}(s))}{d s^2} \\
&  + V^{\tau-u}(s))\frac{d^2 ([\A(t)/\A(u)]^{s}}{d s^2} + 2\frac{d (V^{\tau-u}(s))}{d s} \frac{d ([\A(t)/\A(u)]^{s}}{d s}.
\end{align*}
The first two terms in the RHS above are non-negative as $[\A(t)/\A(u)]^{s}$ and $V^{\tau-u}(s)$ are convex. Further, both are single variable functions and are differentiable. Therefore, their first derivatives are also non-negative and hence the third term in the RHS is also non-negative. Thus, $[\A(t)/\A(u)]^{s} V^{\tau-u}(s)$ is convex.

\end{document}